\documentclass[11pt]{amsart}

\usepackage[marginratio=1:1]{geometry}
\usepackage{color}
\usepackage{graphicx}
\usepackage{amsmath}
\definecolor{refs}{rgb}{0.7,0,0}
\definecolor{ext}{RGB}{112,112,112}
\definecolor{cite}{RGB}{034,113,179}
\usepackage[colorlinks=true,citecolor=cite,linkcolor=refs,urlcolor=ext,backref=page]{hyperref}

\newtheorem{theorem}{Theorem}[section]

\newtheorem{proposition}[theorem]{Proposition}

\theoremstyle{definition}
\newtheorem*{definition}{Definition}

\newcommand{\R}{\mathbb{R}}

\newcommand{\const}{\mathrm{const}}
\newcommand{\spn}{\operatorname{span}}

\newcommand{\BB}{\mathcal{B}}
\newcommand{\MM}{\mathcal{M}}
\newcommand{\FF}{\mathcal{F}}
\newcommand{\TT}{\mathcal{T}}

\title{Deformations of dispersionless Lax systems}
\author{Wojciech Kry\'nski}
\thanks{Partially supported by the grant 2019/34/E/ST1/00188 from the National Science Centre, Poland.}
\address{
Institute of Mathematics, Polish Academy of Sciences, ul. \'Sniadeckich 8, 00-656 Warszawa, Poland}

\email{krynski@impan.pl}

\begin{document}

\begin{abstract}
 We consider dispersionless Lax systems and present a new systematic method of deriving new integrable systems from a given one. We provide examples that include: the dispersionless Hirota equation, the general heavenly equation and the web equations. 
\end{abstract}

\maketitle

\section{Introduction}
In this paper we consider dispersionless integrable systems arising as the integrability condition for a foliation $\FF$ on a bundle $\BB$ with 1-dimensional fibers over a manifold $\MM$. This framework can be adapted to the heavenly equations \cite{Pl,Sch}, the Manakov--Santini system \cite{MS} (see also \cite{DFK}), the Dunajski--Tod equation \cite{DT} (see also \cite{B}), the hyper-CR equation \cite{D1}, the dispersionless Hirota equation \cite{DK} (known also as the abc-equation \cite{Z}), equations related to the $GL(2)$-structures and web geometry \cite{FK2,KP,Kwebs,Kpleb,Khirota,KMet} and many others \cite{C,FK1,KSS,KMakh,MW,MAS}. The quotient space $\TT=\BB/\FF$ is usually referred to as the (real) twistor space and the following double fibration picture arises
\[
\MM\leftarrow \BB\rightarrow \TT.
\]
Our goal is to develop a method of deriving new integrable system  on $\MM$ by equipping the twistor space $\TT$ with additional geometric data.

The twistor methods originate from works of Penorese \cite{P} and Hitchin \cite{H}. All aforementioned systems describe very natural classes of geometric structures on $\MM$. Specifically, these could be anti-self-dual metrics (e.g \cite{B,DT,DFK,Pl}),  Einstein-Weyl structures (e.g. \cite{DFK,DK,FK1,Khirota,MS}), or higher dimensional counterparts like $GL(2)$-structures or Veronese and Kronecker webs (see \cite{FK2,KP,Kwebs,Kpleb,KMet}). However, it is often a subtle and not trivial task to prove generality of the solutions. Moreover, as \cite{KP,MP} shows it is even more difficult to prove whether two equations describe the same set of structures, because the corresponding equations are often non equivalent from the viewpoint of equivalence of PDEs. On the other hand, our method exploits the twistorial picture and automatically gives a correspondence between solutions of different equations. The idea was initially developed in \cite{Khirota} and applied in the specific case of the dispersionless Hirota equation. In the present paper we aim to generalize it and, as an example, provide new applications to the heavenly equations (we refer also to \cite{BFKN, BFKN2, KSS,KM} and \cite{PS}, where another approaches are presented) and the hierarchy of \cite{Kodes} (note that in dimension 4 the structures arise in the context of exotic holonomy groups \cite{Bryant, KMet}).

Different types of deformations of Lax systems were considered before, for instance in \cite{KM} and \cite{KP,PS}. In particular \cite{KM} uses algebraic approach to symmetry algebras, while \cite{KP,PS} exploit normal forms of Nijenhuis operators. A priori the two methods are unrelated, but in the present paper we generalize both of them at the same time. Our main goal, however, is to present an unified framework that can be later used in other context. 

\subsection*{Acknowledgments} I am grateful to Boris Kruglikov for helpful conversations.

\section{General construction}
In this Section we introduce a framework for our studies. Our aim is to define a notion of a deformation of a Lax system. At the end of the Section we formulate Theorem \ref{thm1} that states that the solution space of a deformed equation is in a one to one correspondence with the solution space of the original equation. Once the definitions are introduced, the result is straightforward. However, for applications it would be crucial to determine whether two equations are mutual deformations, or to construct deformations of a given equation. That, in general requires some work, but once the work is done Theorem \ref{thm1} can be applied.   

The present Section ends with additional remarks on the B\"acklund transformations and on a generalization of the construction involving higher order jets.
\subsection{Equations}
Let
\[
\tau\colon U\to \MM
\]
be a vector bundle over a manifold $\MM$ and let
\[
\pi_\MM\colon \BB\to \MM
\]
be a rank-1 fiber bundle over $\MM$. Denote by
\[
\hat\tau\colon\hat U\to\BB
\]
the pullback bundle of $U$. Then, for any section $u$ of $U$, there is a corresponding pullback section $\hat u$ of $\hat U$.
Moreover, let $L_0,\ldots, L_s$ be differential operators acting on sections of $\hat U$ with values in the tangent bundle $T\BB$,
\[
L_i\in \mathrm{Diff}(\hat U,T\BB).
\]
Hence, for a given a section $u$ of $U$, $L_0(\hat u),\ldots, L_s(\hat u)$ are vector fields on $\BB$ depending on $u$ and its derivatives up to certain order. 
\begin{definition}
We say that $u$ is a \emph{solution to the Lax system} defined by $(L_i)_{i=0,\ldots,s}$ if the distribution
\[
D_u=\spn\{L_0(\hat u),\ldots,L_s(\hat u)\}
\]
on $\BB$ is integrable.
\end{definition}

Later on, for simplicity, we shall write $L_i(u)$ instead of $L_i(\hat u)$. The Lax systems are usually defined by a pair $(L_0,L_1)$ of vector fields. The foliation of integral leaves of $D_u$ will be denoted $\FF_u$. For the rest of the paper we shall assume that $U$ and $\BB$ are trivial bundles  $U=\MM\times \R^m$ and $\BB=\MM\times\R P^1$. Moreover, we shall study $\MM$ locally around a given point and hence we will  assume for simplicity that $\MM=\R^n$. The later identification gives the initial coordinates on $\MM$, which will be deformed in the due course.

\subsection{Twistor correspondence}
Let $\TT_u=\BB/\FF_u$ be the space of leaves of foliation $\FF_u$ with the quotient mapping
\[
\pi_{\TT_u}\colon \BB\to \TT_u.
\]

We shall assume that fibers of $\pi_\MM$ intersect fibers of $\pi_{\TT_u}$ transversally, i.e. fibers of $\pi_\MM$ are nowhere tangent to the fibers of $\pi_{\TT_u}$. It follows that one can establish a correspondence between points in $\MM$ and certain curves in $\TT_u$, and conversely a correspondence between points in $\TT_u$ and submanifolds in $\MM$. Indeed, let $x\in \MM$ and denote by $\hat\gamma_x$ a curve in $\BB$ being a counterimage of $x$ with respect to $\pi_\MM$. Let $\gamma_{x,u}=\pi_{\TT_u}(\hat\gamma_x)$. Then $\gamma_{x,u}$ is a curve in $\TT_u$. 
Conversely, let $p\in \TT_u$ and denote by $\hat F_p$ a submanifold in $\BB$ being a counterimage of $p$ with respect to $\pi_{\TT_u}$. Let $F_p=\pi_\MM(\hat F_p)$. Then $F_p$ is a submanifold of $\MM$. For a given solution $u$ we shall denote
\[
\Gamma_u=\{\gamma_{x,u}\ |\ x\in\MM\}.
\] 

Our aim now is to introduce a class of distinguished coordinates on $\MM$ parameterized by solutions $u$ of a Lax system. It is sufficient to find a map from $\Gamma_u$ to $\R^n$ for any solution $u$, because there is a correspondence between points in $\MM$ and $\Gamma_u$. Later on, the codomain $\R^n$ will be assumed to be a fixed space $\tilde\MM=\R^n$.
\begin{definition}
A sequence of corank-1 submanifolds $(T_1, T_2,\ldots,T_m)$ of $\TT_u$ is called a \emph{transversal system} for a family $\Gamma_u$  if any $\gamma\in\Gamma_u$ and $T_i$, $i=1,\ldots,m$, intersect transversally exactly at one point. 
\end{definition}

\begin{definition}
Let $T_u=(T_1,\ldots,T_m)$ be a transversal system of submanifolds of $\TT_u$ for a family $\Gamma_u$ and assume that any $T_l$ possess a fixed local coordinate system $(x_l^1,\ldots,x_l^k)\colon T_l\to\R^k$. Any function $x_l^j$ defines a function on $\Gamma_u$ by a formula
\[
x_l^j(\gamma)=x_l^j(\gamma\cap T_l)
\]
where the right hand side is the value of $x_l^j$ at the intersection point of $\gamma$ and $T_l$ (unique by definition). We say that maps $(x_l^j)$ \emph{induce coordinates} on $\Gamma_u$ if there is subset $(x^1,\ldots,x^n)\subset\{x_l^j\ |\ l=1,\ldots,m,\ j=1,\ldots,k\}$ defining local coordinates on $\Gamma_u$. If this is the case then $\phi_u=(x^1,\ldots,x^n)\colon \Gamma_u\to \R^n$ are called \emph{induced coordinates} on $\MM$.
\end{definition}

\begin{figure}
  \includegraphics[height=180pt]{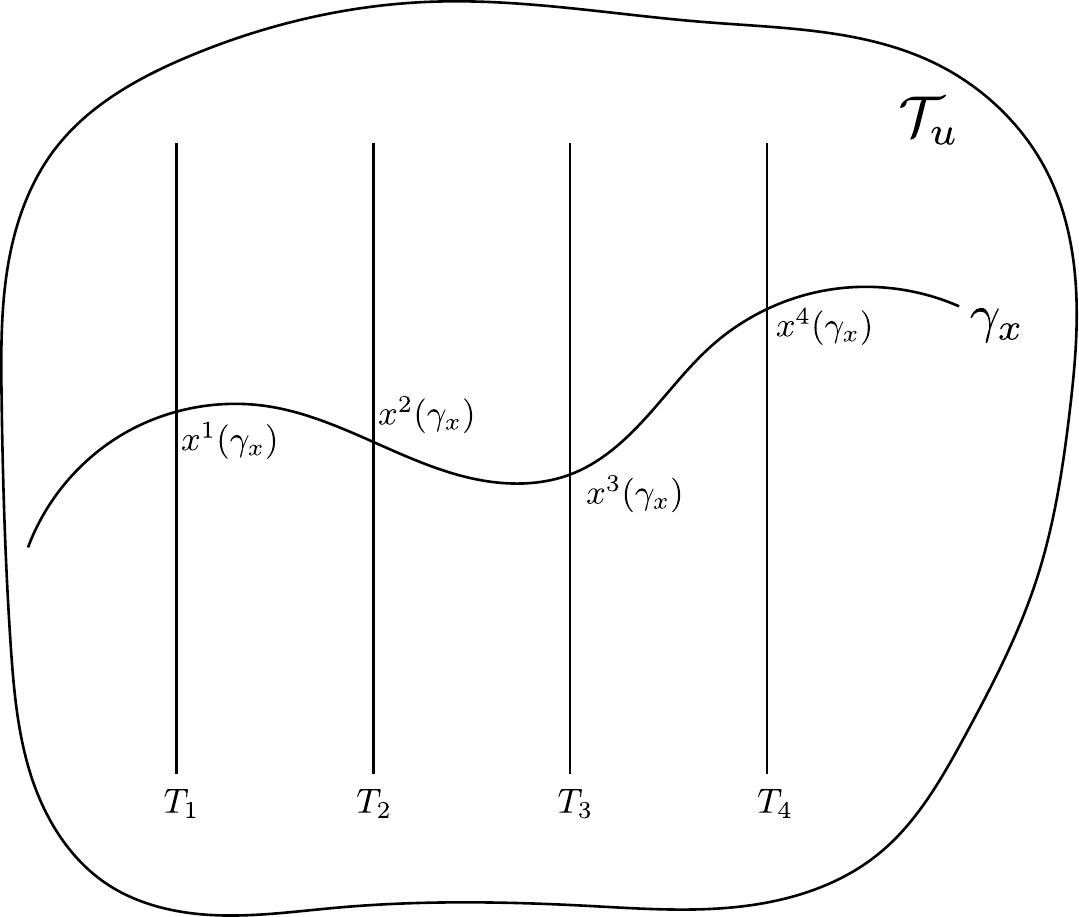}.
  \caption{Transversal system $T_u=(T_1,T_2,T_3,T_4)$ and induced coordinates $(x^i)$ of point $x$ represented by a curve $\gamma_x$ in the twistor space $\TT_u$.}
  \label{fig_transversal1}
\end{figure}

The coordinate functions one gets from the construction are very special. Indeed, we have the following.
\begin{proposition}
If functions $x^i$, $i=1,\ldots,n$, are coordinates on $\MM$ induced by a transversal system on $T_u$, then any foliation $x^i=\const$ on $\MM$ is tangent to the projection of certain integral leaves of distribution $D_u$ via $\pi_\MM$.
\end{proposition}
\begin{proof}
Indeed, a condition $x^i=c$ for some $c\in\R$ fixes a submanifold $S\subset T_s$, where $T_s$ is one of the submanifolds from the transversal system $T_u$. Points $p\in S$ correspond to submanifolds $F_p\subset\MM$ that consist of points represented by curves from $\Gamma_u$ intersecting $S$. On the other hand this submanifolds are, by definition, projections to $\MM$ of integral submanifolds of $D_u$.
\end{proof}

\subsection{Deformations of Lax systems}
Let $u$ be a solution to a Lax system $(L_i)_{i=0,\ldots,s}$ on $\MM$. Let $T_u$ be a transversal system in the corresponding twistor space $\TT_u$ and $\phi_u\colon\MM\to\R^n$ be an induced system of coordinates. From now on we 
assume that all maps $\phi_u$ take values in one fixed copy of $\R^n$ denoted by $\tilde\MM$. To be more precise, we assume that for any solution $u$, and any transversal system $T_u=(T_1,\ldots,T_m)$ the coordinate functions $(x_l^j)_{j=1,\ldots,k}\colon T_l\to \R^k$ take values in a fixed copy of $\R^k$, denoted $\tilde\MM_l$, and $\tilde \MM$ is defined as an image of certain projection from $\tilde\MM_1\times\cdots\times\tilde\MM_k$ to a subspace defined by the choice of $(x^i)$ among $(x_l^j)$.

Having $\tilde\MM$ fixed, we may consider Lax systems $(\tilde L_{i,u})_{i=0,\ldots,s}$ on $\tilde \MM$ defined as pullbacks of the original system through $\phi_u$, for different $u$. These Lax systems are clearly equivalent to the original system and are only written in different coordinate charts. In order to get a significantly new system we shall combine all the systems together. For this we introduce the following definition.

\begin{definition}
A Lax system $(\tilde L_i)_{i=0,\ldots,s}$ on $\tilde\MM$ is a \emph{week deformation} defined by a family of transversal systems $T_u$ and induced coordinates $\phi_u$ of a Lax system  $(L_i)_{i=0,\ldots,s}$ if for any solution $u$ of $(L_i)_{i=0,\ldots,s}$ its pullback
\[
u_*:=u\circ\phi^{-1}_u
\]
 to $\tilde\MM$ via the corresponding $\phi_u$ is a solutions to  $(\tilde L_i)_{i=0,\ldots,s}$. 
\end{definition}

\begin{definition}
A  Lax system $(\tilde L_i)_{i=0,\ldots,s}$ on $\tilde\MM$ is a \emph{deformation} of a Lax system $(L_i)_{i=0,\ldots,s}$ on $\MM$ if the two systems are mutual week deformations of each other and
\[
(u_*)_*=u
\]
holds.
\end{definition}

We get the following tautological result.
\begin{theorem}\label{thm1}
If a Lax system $(\tilde L_i)_{i=0,\ldots,s}$ on $\tilde\MM$ is a deformation of a Lax system on $(L_i)_{i=0,\ldots,s}$ then there is a one to one correspondence between the two solution spaces.
\end{theorem}
\begin{figure}
  \includegraphics[height=150pt]{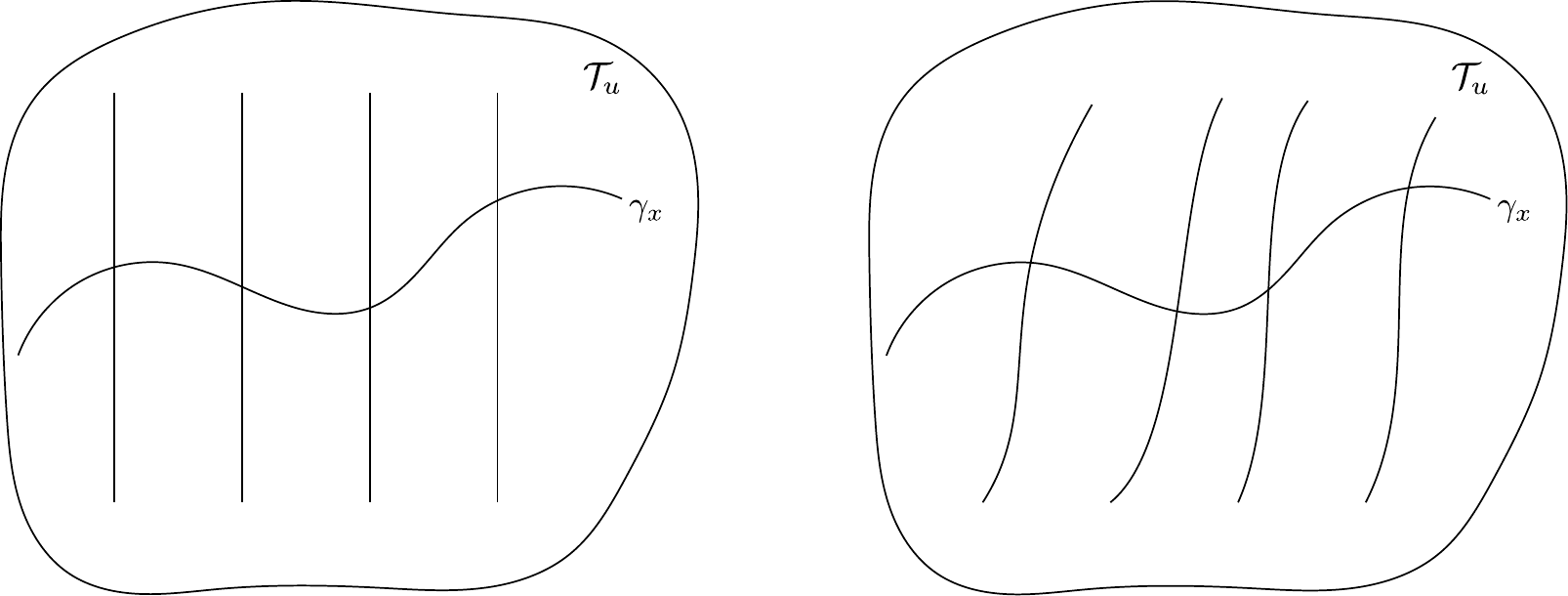}.
  \caption{Two transversal systems on the twistor space $\TT_u$, giving two different coordinate systems on $\MM$.}
  \label{fig_transversal1}
\end{figure}

In the next Section we shall provide a number of examples, but we shall make two general remarks first: concerning the B\"acklund transformations and higher order generalizations.

\subsection{B\"acklund transformations}
The correspondence of Theorem \ref{thm1} does not come in general from contact, or point, transformations of coordinates. Also, in general, it is not a B\"acklund transformation. However, if one assumes that $\MM=\tilde\MM$ then a B\"acklund transformation can be (in principle) found by the following condition
\begin{equation}\label{eq_backlund}
D_u=\tilde D_{\tilde u}
\end{equation}
where $D_u$ and $\tilde D_{\tilde u}$ are distirbutions spanned by the Lax systems $(L_i(u))_{i=0,\ldots,s}$ and $(\tilde L_i(\tilde u))_{i=0,\ldots,s}$ for solutions $u$ and $\tilde u$ respectively.

\subsection{Higher order deformations}
The higher order deformation can be defined by replacing some of $T_i$ in a transversal system by their tangent bundles (or higher order tangent bundles). This can be visualized as a limit, when two (or more) submanifolds $T_i$ in a transversal system are getting arbitrary close. Then, one can introduce coordinates of a curve $\gamma_x$ as a first jet (or higher) at the intersection point of $\gamma_x$ and $T_i$. This approach was applied in the case of the Hirota equation in \cite{Khirota}. We shall not explain details here, but postpone it to separate studies (equations that can be treated in this way include the hyper-CR equation, the Manakov--Santini system and the system that governs the half-integrable Cayley structures introduced in \cite{KMakh}).

\section{Examples}
This section contains a number of examples illustrating general constructions of the previous section. A common feature of all of them is that the twistor space is fibered over a projective space. On the level of Lax pairs it manifests with the lack of a term in the direction of $\partial_\lambda$. In other words, the fiber coordinate $\lambda$ on the bundle $\BB$ descends to a well defined function on $\TT_u$.
\subsection{Dispersionless Hirota equation}
It is proved in  \cite{DK} that the solutions to the dispersionless Hirota equation are in a one to one correspondence with the 3-dimensional hyper-CR Einstein--Weyl structures. The equation, written on manifold $\MM=\R^3$ with coordinates $x^1,x^2,x^3$, reads
\begin{equation}\label{eq_hirota}
au_1u_{23}+bu_2u_{13}+cu_3u_{12}=0,
\end{equation}
where $a,b,c\in\R$ are non-zero constants such that $a+b+c=0$, and $u_i=\partial_iu$. The corresponding Lax pair is of the following form (see \cite{DK})
\begin{equation}\label{eq_L1}
L_0=\partial_3-\frac{u_3}{u_1}\partial_1-\lambda b\partial_3,\qquad L_1=\partial_2-\frac{u_2}{u_1}\partial_1+\lambda c\partial_2,
\end{equation}
where $\lambda$ is an additional affine coordinate on the rank-1 bundle $\BB=\MM\times\R P^1$.
Note that for any value $\lambda$, the two vector fields $L_0$ and $L_1$ span an integrable distribution on $\MM$. Indeed, the Hirota equation is the integrability condition. The corresponding conformal metric $[g]$ and the Weyl connection can be found in \cite{DK} (formula (4)).

Let $\lambda_1, \lambda_2,\lambda_3$ be such that $a=\lambda_2-\lambda_3$, $b=\lambda_3-\lambda_1$ and $c=\lambda_1-\lambda_2$. Performing a M\"obius transformation $\lambda\mapsto\frac{1}{\lambda_1-\lambda}$ one puts the Lax pair \eqref{eq_L1} in the form
\begin{equation}\label{eq_L2}
L_0=(\lambda-\lambda_1)\frac{u_3}{u_1}\partial_1-(\lambda-\lambda_3)\partial_3,\qquad L_1=(\lambda-\lambda_1)\frac{u_2}{u_1}\partial_1-(\lambda-\lambda_2)\partial_2,
\end{equation}
which clearly gives the same Lax equation for $u$. One checks that for $\lambda=\lambda_i$ the distribution spanned by $L_0$ and $L_1$  is tangent to the foliation $x^i=\const$.

It follows that coordinates $(x^1,x^2,x^3)$ can be interpret as induced coordinates from a transversal system. Indeed, $\TT_u=\BB/D_u$ is the space of all integral manifolds of $D_u=\spn\{L_0,L_1\}$. Moreover, $\TT_u$ is two-dimensional and, as mentioned before, $\lambda$ is a well defined function on $\TT$ because it is constant on leaves of $D_u$. The transversal system is defined by submanifolds
\[
T_i=\{p\in\TT_u\ |\ \lambda(p)=\lambda_i\}.
\]
Function $x^i$ is a diffeomorphism from $T_i$ to $\R$, i.e. it parameterizes leaves of the aforementioned foliation $x^i=\const$ on $\MM$. 

Let us replace now $T_i$ by general submanifolds of $\TT_u$ of the form
\[
\tilde T_i=\{p\in\TT_u\ |\ f_i(p)=0\},	
\]
for some functions $f_i\colon\TT_u\to\R$  on the twistor space (with $0$ being a regular value). It is proved in \cite{Khirota} that any choice of coordinates $x^i$ on $\tilde T_i$ transforms \eqref{eq_hirota} into 
\begin{equation}\label{eq_hirota2}
(\lambda_2(x^2)-\lambda_3(x^3))u_1u_{23}+(\lambda_3(x^3)-\lambda_1(x^1))u_2u_{13}+(\lambda_1(x^1)-\lambda_2(x^2))u_3u_{12}=0,
\end{equation}
where now each $\lambda_i(x^i)$, $i=1,2,3$, is a function of one variable.  Precisely, $\lambda_i$ is defined as the restriction of function $\lambda$ from $\TT_u$ to submanifold $\tilde T_i$, and $x^i$ is a coordinate function on $\tilde T_i$, which consequently becomes an induced coordinate on $\MM$. In the previous case this restriction of $\lambda$ to $T_i$ is just a constant function.

Equation \eqref{eq_hirota2} was derived for the first time in \cite{KP} where the authors analyzed the so-called Nijenhuis operators associated to Veronese webs. The general assumption in their approach is that the Nijenhuis operators have no singular points which is reflected in the condition  $\partial_{x^i}\lambda_i(x^i)\neq 0$. In this case a simple coordinate change (a point transformation) gives $\lambda_i(x^i)=x^i$. In our approach $\lambda_i(x^i)$ can be arbitrary smooth function of one variable. In particular we admit $\partial_{x^i}\lambda_i(x^i)=0$ for certain values of $x^i$.

In \cite{KP} it is also proved that equations with constant coefficients $\lambda_i$ are contactly non-equivalent to the one with non-constant $\lambda_i$ (one can check that the corresponding symmetry groups are different). Note that the corrdinate change between coordinate systems induced by different transversal systems does not establish a contact equivalence of the Lax systems. The reason is that the coordinate change depends on a given solution $u$. 

On the other hand, formula \eqref{eq_backlund} can be applied and it descents to a B\"acklund equivalence, which in the present case was found in \cite{KP} as
\[
\lambda_i\tilde\lambda_ju_i\tilde u_j=\lambda_j\tilde\lambda_iu_j\tilde u_i,\qquad i,j=1,2,3.
\] 
where $u$ and $\tilde u$ are correspondingly solutions to \eqref{eq_hirota2} or to a variant of \eqref{eq_hirota2} with functions $\lambda_i$ replaced by $\tilde\lambda_i$.

\subsection{Higher dimensional Veronese webs}
Veronese webs are higher dimensional counterparts of the hyper-CR Einstein--Weyl structures described by equation \eqref{eq_hirota}. We refer to \cite{DK} for the proof of the 3-dimensional correspondence between the Veronese webs and the Einstein--Weyl structures. The general case was studied in \cite{Kodes} where it is additionally proved that the webs underlie very particular integrable paraconformal structures (so called totally geodesic $GL(2)$-geometries).

The Veronese webs on $\MM=\R^n$ are 1-parameter families of corank-one foliations \cite{JK,Kwebs,Z} such that the annihilating 1-forms $\omega_\lambda$ give rise to a rational normal curves $\lambda\mapsto\R\omega_\lambda(x)\in P(T^*_x\MM)$ at each point $x\in\MM$. The curves replace cones of null directions of $[g]$ in the 3-dimensional case. It is proved in \cite{Kodes} that the Veronese webs are in a one to one correspondence with solutions to the following system
\begin{equation}\label{eq_hirotasys}
(\lambda_i-\lambda_j)u_ku_{ij}+(\lambda_k-\lambda_i)u_ju_{jk}+(\lambda_j-\lambda_k)u_ku_{ij}=0,
\end{equation}
where $i,j,k=1,\ldots,n$. Indeed, the equations are equivalent to the integrability condition
\[
d\omega_\lambda\wedge\omega_\lambda=0
\]
where $\omega_\lambda$ is given by
\begin{equation}\label{eq_omega}
\omega_\lambda=\sum_{i=1}^n\prod_{j\neq i}(\lambda-\lambda_j)u_idx^i.
\end{equation}

System \eqref{eq_hirotasys} appeared in \cite{DK} for the first time and turned out to be a Lax system with
\begin{equation}\label{eq_L3}
L_i=(\lambda-\lambda_1)\frac{u_i}{u_1}\partial_1-(\lambda-\lambda_i)\partial_i,\qquad i=1,\ldots,n.
\end{equation}
The corresponding twistor space is two dimensional and the most natural transversal system can be defined as before
\[
T_i=\{p\in\TT\ |\ \lambda(p)=\lambda_i\},
\]
where now $i=1,\ldots,n$. Note that, as in the 3 dimensional case, $\lambda$ is constant along vector fields $L_i$ independently of $u$, and therefore can be treated as a function on the twistor space  $\TT_u$.

Deforming $T_i$ as in the 3-dimensional case we get that the constants $\lambda_i$ can be replaced by arbitrary (smooth) functions of one variable $\lambda_i=\lambda_i(x_i)$. Indeed we have
\[
(\lambda_i(x^i)-\lambda_j(x^j))u_ku_{ij}+(\lambda_k(x^k)-\lambda_i(x^i))u_ju_{jk}+(\lambda_j(x^j)-\lambda_k(x^k))u_ku_{ij}=0,
\]
$i,j,k=1,\ldots,n$, as an integrability condition for $\omega_\lambda$ given by \eqref{eq_omega} with constants $\lambda_i$ replaced by functions. Note that, if $\partial_{x^i}\lambda_i\neq 0$ then coordinate change (a point transformation) reduces $\lambda_i(x^i)=x^i$.

Finally, applying \eqref{eq_backlund} we find B\"acklund transformations given by
\[
\lambda_i\tilde\lambda_ju_i\tilde u_j=\lambda_j\tilde\lambda_iu_j\tilde u_i,\qquad i,j=1,\ldots,n.
\] 
where $\lambda$ and $\tilde\lambda_i$ are functions as before.

\subsection{Heavenly equations}
In a recent paper \cite{KSS} an interesting approach to the heavenly equations is presented. In this context $\MM=\R^4$ and one looks for split-signature self-dual Ricci flat metrics. It is well known that metrics of this type are in a one to one correspondence with solutions to the (first) Pleba\'nski equation
\begin{equation}\label{eq_plebanski}
u_{13}u_{24}-u_{14}u_{23}=1
\end{equation}
which has a Lax pair on $\BB=\MM\times\R P^1$ of the following form
\[
L_0=-\partial_3+\lambda(u_{13}\partial_2-u_{23}\partial_1),\quad L_1=-\partial_4+\lambda(u_{14}\partial_2-u_{24}\partial_1)
\]
The authors of \cite{KSS} utilize the eigenfunctions as coordinates, where eigenfunctions are understood as solutions to linear system
\begin{equation}\label{eq_eigen}
L_i|_{\lambda=\lambda^*}\psi=0, \qquad i=0,1,
\end{equation}
for unknown function $\psi$, where $\lambda^*$ is a fixed value of $\lambda$. It follows that any solution $\psi$ is a function on a submanifold of $\TT_u$ defined as
\[
T_{\lambda^*}=\{p\in\TT_u\ |\ \lambda(p)=\lambda^*\},
\]
where as in the previous examples $\lambda$ is treated here as a function on $\TT_u$. Conversely, any function on $T_{\lambda^*}$ is a solution to \eqref{eq_eigen}. 

In the present case $\TT_u$ is 3-dimensional, and consequently $T_{\lambda^*}$ is a 2-dimensional surface in $\TT$. It follows that there are two functionally independent functions on each $T_{\lambda^*}$. \cite{KSS} uses the functions as coordinates on $\MM$. In our terminology these are induced coordinates for a transversal system. For the Pleba\'nski equation it is enough to take two distinct values of $\lambda^*$, say $\lambda_1$ and $\lambda_2$, and the corresponding  transversal system $T_u=(T_1,T_2)$. We get four functions $x^i_l$, $i,l=1,2$ in this way, as needed. 

However, one can take $T_u=(T_1,T_2,T_3, T_4)$ for four different values of $\lambda^*=\lambda_l$, $l=1,2,3,4$. Then coordinates on $T_l$ give  8 functions $x^i_l$ in total. We pick 4 of them, one for each $\lambda_l$. It is proved in \cite{KSS} (Theorem 4.1) that as a result one gets the general heavenly equation (see also \cite{Sch})
\begin{equation}\label{eq_heavenly}
(\lambda_1-\lambda_2)(\lambda_3-\lambda_4)u_{12}u_{34}+(\lambda_2-\lambda_3)(\lambda_1-\lambda_4)u_{23}u_{14}+(\lambda_3-\lambda_1)(\lambda_2-\lambda_4)u_{31}u_{24}=0.
\end{equation}
The corresponding Lax pair has the following form
\[
\begin{aligned}
&L_0=(\lambda-\lambda_1)(\lambda_2-\lambda_3)u_{23}\partial_1+(\lambda-\lambda_2)(\lambda_3-\lambda_1)u_{13}\partial_2+(\lambda-\lambda_3)(\lambda_1-\lambda_2)u_{12}\partial_3,\\
&L_1=(\lambda-\lambda_1)(\lambda_2-\lambda_4)u_{24}\partial_1+(\lambda-\lambda_2)(\lambda_4-\lambda_1)u_{14}\partial_2+(\lambda-\lambda_4)(\lambda_1-\lambda_2)u_{12}\partial_4.
\end{aligned}
\]

Now, the transversal system $T_u$ can be deformed to a quadruple $\tilde T_u=(\tilde T_1,\tilde T_2,\tilde T_3, \tilde T_4)$ of arbitrary surfaces (transversal with respect to the family $\Gamma_u$) in $\TT_u$. Similarly to the Hirota equation such a deformation replaces each $\lambda_i$ by a smooth function which in general can be written in the form
\begin{equation}\label{eq_generalL}
\lambda_i=\lambda_i(x_i^1,x_i^2)
\end{equation}
where $(x_i^1,x_i^2)$ are two coordinate functions on $\tilde T_i$. One of them, say $x_i^1$, becomes a coordinate on $\MM$ (denote it by $x^i$). Denote by $\phi^i$ the pullback of the second coordinate function (i.e. $x_i^2$) from $\tilde T_i$ to $\MM$. Then $\phi^i=\phi^i(x^1,x^2,x^3,x^4)$ is a general function satisfying \eqref{eq_eigen} for $\lambda^*=\lambda_i(x^1_i,x^2_i)$. 


The Lax pair in the new coordinate system reads
\[
\begin{aligned}
&\tilde L_0=(\lambda-\lambda_1)(\lambda_2-\lambda_3)\tilde u_{23}\partial_1+(\lambda-\lambda_2)(\lambda_3-\lambda_1)\tilde u_{13}\partial_2+(\lambda-\lambda_3)(\lambda_1-\lambda_2)\tilde u_{12}\partial_3,\\
&\tilde L_1=(\lambda-\lambda_1)(\lambda_2-\lambda_4)\tilde u_{24}\partial_1+(\lambda-\lambda_2)(\lambda_4-\lambda_1)\tilde u_{14}\partial_2+(\lambda-\lambda_4)(\lambda_1-\lambda_2)\tilde u_{12}\partial_4.
\end{aligned}
\]
where $\tilde u_{ij}$ are, a priori, certain expressions involving second order derivatives of the original $u$ and of the coordinate change. However, it turns out that $\tilde u_{ij}=\partial_i\partial_j\tilde u$ are second order derivatives (with respect to the new coordinates) of a function $\tilde u$ provided that new coordinates, as well as functions $\phi^i$, are defined as functions on $\tilde T_i$, or equivalently they are solutions to
\begin{equation}\label{eq_eigen2}
\tilde L_0|_{\lambda=\lambda_i(x^1,\phi^i)}\psi=0, \qquad \tilde L_1|_{\lambda=\lambda_i(x^1,\phi^i)}\psi=0,
\end{equation}
in new coordinates. The integrability condition for the new Lax pair gives \eqref{eq_heavenly} (with $\lambda_i$ upgraded to functions).  Notice, that there is one obvious solution to \eqref{eq_eigen2} witch is $\phi^i=\partial_i\tilde u=\tilde u_i$ (written in new coordinates). Summarizing, we get the following equation
\[
\begin{aligned}
&(\lambda_1(x^1,u_1)-\lambda_2(x^2,u_2))(\lambda_3(x^3,u_3)-\lambda_4(x^4,u_4))u_{12}u_{34}\\
&+(\lambda_2(x^2,u_2)-\lambda_3(x^3,u_3))(\lambda_1(x^1,u_1)-\lambda_4(x^1_4,u_4))u_{23}u_{14}\\
&+(\lambda_3(x^3,u_3)-\lambda_1(x^1,u_1))(\lambda_2(x^2,u_2)-\lambda_4(x^4,u_4))u_{31}u_{24}=0.
\end{aligned}
\]

Specific cases of this particular deformation appeared independently in \cite{KM} and \cite{PS}. In \cite{KM} authors consider deformations of the symmetry algebra and their algebraic procedure give our deformed equation but only with respect to $\lambda_4$ and coordinate $x^4$ (it is case $(I)$ of \cite{KM}). Function $Q$ in \cite{KM} is given by
\[
Q(x^4,u_4)=\frac{(\lambda_1-\lambda_3)(\lambda_2-\lambda_4(x^4,u_4))}{(\lambda_2-\lambda_3)(\lambda_1-\lambda_4(x^4,u_4))}.
\]
On the other hand authors of \cite{PS} apply approach of the Nijenhuis operators and get $\lambda_i=\lambda_i(x^i)$. This corresponds to the case when the foliation  $x^1_i=\const$ on $\tilde T_i$ coincide with the foliation $\lambda=\const$ on $\TT_u$ restricted to $\tilde T_i$. As in the previous examples our method is more general also in this case, since it covers singularities of the Nijenhuis operators - corresponding to points where $\partial_{x^i}\lambda_i(x^i)=0$.


\end{document}